\documentclass[conference]{IEEEtran}
\usepackage{amsthm, amssymb, amsmath}
\usepackage{verbatim, float}
\usepackage{mathrsfs}
\usepackage{graphics}
\usepackage[usenames,dvipsnames]{pstricks}
\usepackage{epsfig}
\usepackage{pst-grad} 
\usepackage{pst-plot} 
\usepackage{cases}
\usepackage{algorithmic}
\usepackage{bm}

\usepackage[top=1.4cm, bottom=1.4cm, left=1.33cm, right=1.33cm]{geometry}
\usepackage{url}

\usepackage{hhline}

\usepackage{etoolbox}
\makeatletter
\patchcmd{\@makecaption}
  {\scshape}
  {}
  {}
  {}
\makeatother

\usepackage[singlelinecheck=false]{caption}
\usepackage{diagbox}

\theoremstyle{plain}
\newtheorem{theorem}{Theorem}
\newtheorem{lemma}{Lemma}

\newtheorem{corollary}{Corollary}
\newtheorem{proposition}{Proposition}

\theoremstyle{definition}
\newtheorem{definition}{Definition}

\newcommand{\C}{{\mathcal C}}

\newcommand{\ba}{{\boldsymbol a}}

\newcommand{\bO}{{\boldsymbol 0}}

\newcommand{\ft}{\mathbb{F}_2}
\newcommand{\fq}{\mathbb{F}_q}

\newcommand{\fqt}{\mathbb{F}_{q^t}}

\newcommand{\fqts}{\mathbb{F}_{q^{t-s}}}

\newcommand{\be}{\beta}

\newcommand{\tr}{\mathsf{Tr}}

\newcommand{\rank}{\mathsf{rank}_B}
\newcommand{\rankt}{\mathsf{rank}_2}
\newcommand{\rankq}{\mathsf{rank}_q}

\newcommand{\define}{\stackrel{\mbox{\tiny $\triangle$}}{=}}

\newcommand{\et}{{\emph{et al.}}}

\newcommand{\Cd}{\mathcal{C}^\perp}
\newcommand{\rsk}{\text{RS}(A,k)}
\newcommand{\grskl}{\text{GRS}(A,k,\boldsymbol{\lambda})}
\newcommand{\grsnkl}{\text{GRS}(A,n-k,\boldsymbol{\lambda})}

\newcommand{\fa}{f(\alpha)}
\newcommand{\fas}{f(\alpha^*)}

\newcommand{\ga}{g(\alpha)}

\newcommand{\gix}{g_i(x)}
\newcommand{\gia}{g_i(\alpha)}
\newcommand{\gias}{g_i(\alpha^*)}

\newcommand{\gox}{g_1(x)}
\newcommand{\goa}{g_1(\alpha)}
\newcommand{\goas}{g_1(\alpha^*)}

\newcommand{\gtx}{g_t(x)}
\newcommand{\gta}{g_t(\alpha)}
\newcommand{\gtas}{g_t(\alpha^*)}

\newcommand{\gsx}{h^*(x)}

\newcommand{\gsas}{h^*(\alpha^*)}

\begin{document}

\title{Optimal Repair Schemes for Some Families of Full-Length Reed-Solomon Codes
\thanks{
H. Dau and O. Milenkovic are with the Coordinated Science Laboratory, University of Illinois at Urbana-Champaign, 1308 W. Main Street, Urbana, IL 61801, USA. Emails: \{hoangdau, milenkov\}@illinois.edu.
}
}
\author{
  \IEEEauthorblockN{
    Hoang~Dau~
    and
    Olgica~Milenkovic
    }
  {\normalsize
    \begin{tabular}{ccc}
      ECE Department, University of Illinois at Urbana-Champaign \\
      Emails: hoangdau@illinois.edu, milenkov@illinois.edu
    \end{tabular}}\vspace{-3ex}
    }
\date{}
\maketitle

\begin{abstract}
Reed-Solomon codes have found many applications in practical storage systems, but were until recently considered unsuitable for distributed
storage applications due to the widely-held belief that they have poor repair bandwidth. The work of Guruswami and Wootters (STOC'16) has shown that one can actually perform bandwidth-efficient linear repair with Reed-Solomon codes: When the codes are over the field $\fqt$ and the number of parities $r \geq q^s$, where $(t-s)$ \emph{divides} $t$, there exists a linear scheme that achieves a repair bandwidth of $(n-1)(t-s)\log_2 q$ bits. We extend this result by showing the existence of such a linear repair scheme for \emph{every} $1 \leq s < t$. 
Moreover, our new schemes are optimal among all linear repair schemes for Reed-Solomon codes when $n = q^t$ and $r = q^s$.
Additionally, we improve the lower bound on the repair bandwidth for Reed-Solomon codes, also established in the work of Guruswami and Wootters. 
\end{abstract}

\section{Introduction}
\label{sec:intro}

The \emph{repair bandwidth} is an important parameter of erasure codes used for assessing their performance in distributed storage applications~\cite{Dimakis_etal2007, Dimakis_etal2010}.
In distributed storage systems, for a chosen finite field $F$, a data vector in $F^k$ is mapped to a codeword vector in $F^n$, whose entries are 
subsequently stored at different storage nodes. When a node fails, the symbol stored at that node is erased (lost). 
A replacement node has to recover the content stored at the failed node by downloading information from the remaining operational nodes. 
The repair bandwidth is the total amount of information that the replacement node
has to download in order to successfully complete the repair process. 

At first glance, Reed-Solomon codes~\cite{ReedSolomon1960,MW_S} do not appear suitable for repair tasks as recovering the content stored at a single failed node requires downloading the whole file, i.e., $k$ \emph{symbols} over $F$. To mitigate this problem, a number of repair-efficient codes such as regenerating codes~\cite{Dimakis_etal2007, Dimakis_etal2010, Dimakis_etal2010_survey} and locally repairable codes~\cite{OggierDatta2011,GopalanHuangSimitciYekhanin2012,PapailiopoulosDimakis2012} were constructed and proposed for practical implementation.
Reed-Solomon codes, along with replication codes, nonetheless remain the most frequently used erasure codes. They are core components of  storage systems such as Google File System II, Quantcast File
System, Yahoo Object Store, Facebook HDFS-RAID, and HDFS-EC.  

Guruswami and Wootters~\cite{GuruswamiWootters2016} recently proposed a bandwidth-efficient repair method Reed-Solomon codes. The key idea behind their method is to download \emph{sub-symbols} rather than \emph{symbols}. More precisely, a single erased symbol is recovered by collecting a sufficient number of its (field) traces, each of which can be constructed from a number of traces of other symbols. 
As all traces belong to a smaller subfield $B = \fq$ of $F = \fqt$
and traces from the same symbol are related, the repair bandwidth can be significant reduced. 
One instance of interest for which their method produces a repair scheme with \emph{optimal} repair bandwidth is when the code has ``full'' length $n = |F| = q^t$ and redundancy $r \define n - k = q^s$, where $(t-s)|t$. 

On the other hand, there exists a rich literature on non-Reed-Solomon maximum-distance separable (MDS) codes (see~\cite[Chp.~10]{MW_S}) with optimal repair bandwidth, referred to as 
minimum-storage regenerating (MSR) codes~\cite{Dimakis_etal2010_survey, WikiEC}. 
Some notable examples of MSR codes include the low-rate product-matrix code~\cite{RashmiShahKumar2011} and all-rates Zigzag code~\cite{TamoWangBruck2011, WangTamoBruck2011}. 
High-rate MSR codes, however, employ \emph{subpacketization} levels that are often exponential in $n$. The problem of finding bandwidth-efficient repair schemes for MDS codes such 
as Reed-Solomon codes was first raised by Dimakis {\et}~\cite{Dimakis_etal2010_survey}. 
The repair process for Reed-Solomon codes was studied in the work of Shanmugam {\et}~\cite{Shanmugam2014}, which proposed interference alignment techniques to repair systematic node failures; and Guruswami and Wootters~\cite{GuruswamiWootters2016}, which introduced the trace collection technique to repair any single erasure for Reed-Solomon codes. Their technique was recently generalized to tackle two and three erasures by Dau {\et}~\cite{DauDuursmaKiahMilenkovic2016}. 
In a related line of work, Ye and Barg~\cite{YeBarg_ISIT2016} constructed Reed-Solomon codes with asymptotically optimal repair bandwidth (among all MDS codes) and exponentially large subpacketization. The work of~\cite{GuruswamiWootters2016} and our recent results~\cite{DauDuursmaKiahMilenkovic2016}, in contrast, use subpacketization as small as $\log_q n$. 

Repair schemes for Reed-Solomon codes with optimal bandwidths have only been constructed for redundancies equal to $r = q^s$, where $(t-s)$ divides $t$. The goal of this work is to determine the optimal repair bandwidths of \emph{all} full-length Reed-Solomon codes for which the number of parities $r$ is allowed to vary from 
$1$ to $n-1$. For this scenario, we settle an important case: we present schemes with repair bandwidth $(n-1)(t-s)\log_2 q$, for $r \geq q^s$ and for \emph{every} $s < t$. This bandwidth is optimal for Reed-Solomon codes whenever $n = q^t$ and $r = q^s$. The key idea behind our approach is to use linearized, instead of trace polynomials, to generate the dual codewords used for repair.  Additionally, we derive a lower bound on the repair bandwidth of Reed-Solomon codes that improves the bound of~\cite{GuruswamiWootters2016}. Theoretical results and numerical evidence suggest that the lower bound matches the optimal repair bandwidths of all full-length Reed-Solomon codes.  

The paper is organized as follows. We provide relevant definitions and introduce the terminology used throughout the paper and then proceed to 
discuss the Guruswami-Wootters repair scheme for Reed-Solomon codes in Section~\ref{sec:GW}. 
The improved lower bound on the repair bandwidth is presented in Section~\ref{sec:lower_bound}. 
The main results of the work are presented in Section~\ref{sec:main}.

\section{Repairing One Erasure in Reed-Solomon Codes}
\label{sec:GW}

We start by introducing relevant definitions and the notation used in all subsequent derivations, and then proceed to 
review the approach proposed by Guruswami and Wootters~\cite{GuruswamiWootters2016} for repairing a single erasure/node failure in Reed-Solomon codes. 

\subsection{Definitions and Notation}
\label{subsec:def_not}

Let $[n]$ denote the set $\{1,2,\ldots,n\}$. Let $B = \fq$ be the finite field of $q$ elements, for some prime power $q$. Let $F = \fqt$ be an extension field of $B$, where $t \geq 1$. 
We refer to the elements of $F$ as \emph{symbols} and the elements of $B$ as \emph{sub-symbols}. The field $F$ may also be viewed as a vector space of dimension $t$ over $B$, i.e. $F \cong B^t$, and hence each symbol in $F$ may be represented 
as a vector of length $t$ over $B$. A linear $[n,k]$ code $\C$ over $F$ is a subspace of $F^n$ of dimension $k$. Each element of a code is referred to as a codeword. 
The dual of a code $\C$, denoted $\Cd$, is the orthogonal complement of $\C$, and has dimension $r = n - k$. 
 
\begin{definition} 
\label{def:RS}
Let $F[x]$ denote the ring of polynomials over $F$. A Reed-Solomon code $\rsk \subseteq F^n$ of dimension $k$ over a finite
field $F$ with evaluation points $A=\{\alpha_1,\alpha_2,\ldots, \alpha_n\}\subseteq F$
is defined as: 
\[
\rsk = \Big\{\big(f(\alpha_1),\ldots,f(\alpha_n)\big) \colon f \in F[x], \deg(f) < k \Big\}. 
\]
\end{definition} 
A \emph{generalized} Reed-Solomon code, $\grskl$, where $\boldsymbol{\lambda} = (\lambda_1,\ldots,\lambda_n)\in F^n$, is defined similarly to a Reed-Solomon code, except that the codeword
corresponding to a polynomial $f$ is defined as $\big( \lambda_1f(\alpha_1),\ldots,\lambda_n f(\alpha_n) \big)$, $\lambda_i \neq 0$ for all $i \in [n]$. 
It is well known that the dual of a Reed-Solomon code $\rsk$, for any $n \leq |F|$, is a generalized Reed-Solomon code $\grsnkl$, for some multiplier vector $\boldsymbol{\lambda}$~(see~\cite[Chp.~10]{MW_S}). Whenever clear from the context, we use $f(x)$ to denote a polynomial of degree at most $k-1$, which corresponds to a codeword of the Reed-Solomon code $\C=\rsk$, and $g(x)$ to denote a polynomial of degree at most $r-1=n-k-1$, which corresponds to a codeword of the dual code $\Cd$. Since
$
\sum_{\alpha \in A}\ga(\lambda_{\alpha}\fa) = 0, 
$
we refer to the polynomial $g(x)$ as a check polynomial for $\C$. 
Note that when $n = |F|$, we have $\lambda_\alpha = 1$ for all $\alpha \in F$. 
In general, as recovering $\fa$ is equivalent to recovering $\lambda_{\alpha}\fa$, to simplify the notation, we omit the factor $\lambda_\alpha$ in our derivations.

\subsection{The Guruswami-Wootters Repair Scheme for One Erasure}
\label{subsec:GW}

Suppose that the polynomial $f(x) \in F[x]$ corresponds to a codeword of the Reed-Solomon code $\C=\rsk$ and that $f(\alpha^*)$ is the erased 
symbol, where $\alpha^* \in A$ is an evaluation point of the code. 

Given that $F$ is an extension field of $B$ of degree $t$, i.e. $F = \fqt$ and $B = \fq$, for some prime power $q$, one may define the field trace of any symbol $\alpha \in F$ according to
$\mathsf{Tr}_{F/B}(\alpha) = \sum_{i = 0}^{t-1} \alpha^{q^i}$. The trace belongs to $B$. When clear from the context, we omit the subscript $F/B$. 
The key points in the repair scheme proposed by Guruswami and Wootters~\cite{GuruswamiWootters2016} can be 
summarized as follows. First, each symbol in $F$ can be recovered from its $t$ independent traces. More precisely, given a basis $u_1,u_2,\ldots,u_t$ of $F$ over $B$, any $\alpha \in F$ can be uniquely determined given the values of $\tr(u_i\,\alpha)$ for $i\in [t]$, i.e. $\alpha = \sum_{i=1}^t\tr(u_i \alpha)u^\perp_i$, where $\{u^\perp_i\}_{i=1}^t$ is the dual (trace-orthogonal) basis of $\{u_i\}_{i=1}^t$ (see, for instance~\cite[Ch.~2, Def.~2.30]{LidlNiederreiter1986}).
Second, when $r \geq q^{t-1}$, the trace function also provide checks that generate repair equations with coefficients that are linearly dependent over $B$, which keeps the repair cost low. 

Note that the checks of $\C$ are precisely those polynomials $g(x) \in F[x]$ that satisfy $\deg(g) \leq r-1$. 
For $r \geq q^{t-1}$, we may define repair checks via the trace function as follows. For each $u \in F$
and $\alpha \in F$, we introduce the polynomial \vspace{-5pt}
\begin{equation} 
\label{eq:p}
g_{u,\alpha}(x) = \tr\big(u(x-\alpha)\big)/(x-\alpha). \vspace{-5pt}
\end{equation} 
By the definition of a trace function, the next lemma follows in a straightforward manner. 

\begin{lemma}[\cite{GuruswamiWootters2016}]
\label{lem:p}
The polynomial $g_{u,\alpha}(x)$ defined in~\eqref{eq:p} satisfies the following properties.\\
\quad (a) $\deg(g_{u,\alpha}) = q^{t-1}-1$;\quad \text{(b)} $g_{u,\alpha}(\alpha) = u$.
\end{lemma}

By Lemma~\ref{lem:p}~(a), $\deg(g_{u,\alpha}) = q^{t-1}-1 \leq r-1$. Therefore, the polynomial $g_{u,\alpha}(x)$ corresponds to a codeword of $\Cd$ and is a check for $\C$. 
Now let $U = \{u_1,\ldots,u_t\}$ be a basis of $F$ over $B$, and set \vspace{-5pt}
\[
\gix \define g_{u_i,\alpha^*}(x) = \tr\big(u_i(x-\alpha^*)\big)/(x-\alpha^*), \quad i \in [t]. \vspace{-5pt}
\]  
These $t$ polynomials correspond to $t$ codewords of $\Cd$. 
Therefore, we obtain $t$ equations of the form  \vspace{-5pt}
\begin{equation} 
\label{eq:GW_repair}
\gias\fas = - \sum_{\alpha \in A \setminus \{\alpha^*\}} \gia\fa, \quad i \in [t].
\vspace{-5pt} 
\end{equation} 
A key step in the Guruswami-Wootters repair scheme is to apply the trace function to
both sides of \eqref{eq:GW_repair} to obtain $t$ different \emph{repair equations} \vspace{-5pt}
\begin{equation} 
\label{eq:GW_repair_trace}
\tr\big(\gias\fas\big) = - \sum_{\alpha \in A \setminus \{\alpha^*\}} \tr\big(\gia\fa\big), \ i \in [t]. \vspace{-5pt}
\end{equation}
According to Lemma~\ref{lem:p}~(b), $\gias = u_i$, for $i = 1,\ldots,t$. 
Moreover, by the linearity of the trace function, we can rewrite~\eqref{eq:GW_repair_trace} as follows. For $i = 1,\ldots,t$, \vspace{-5pt}
\begin{equation} 
\label{eq:GW_repair_trace_explicit}
\tr\big(u_i\fas\big) = - \sum_{\alpha \in A \setminus \{\alpha^*\}} \tr\big(u_i(\alpha -\alpha^*)\big) \times \tr\Big(\dfrac{f(\alpha)}{\alpha-\alpha^*}\Big).
\vspace{-5pt} 
\end{equation}
The right-hand side sums of the equations~\eqref{eq:GW_repair_trace_explicit} may be computed by downloading the reconstruction trace 
$\tr\Big(\frac{f(\alpha)}{\alpha-\alpha^*}\Big)$ from the node storing $\fa$, for each
$\alpha \in A \setminus \{\alpha^*\}$.
As a consequence, the $t$ independent traces $\tr\big(u_i\fas\big)$, $i = 1,\ldots,t$, of $\fas$ can be determined by downloading one sub-symbol from each of the $n-1$ available nodes, and the erased symbol $\fas$ may subsequently be recovered from its $t$ independent traces. The following theorem summarizes this brief discussion. 

\begin{theorem}[\cite{GuruswamiWootters2016}] 
\label{thm:GW}
When $r \geq q^{t-1}$, there exists a repair scheme for Reed-Solomon codes with a repair bandwidth of $(n-1)\log_2 q$ bits. More generally, if $r \geq q^s,$ where $(t-s)|t$, by setting $B = \fqts$ one can devise a repair scheme with bandwidth $(n-1)(t-s)\log_2 q$ bits, which is optimal
when $n = q^t$ and $r = q^s$. 
\end{theorem} 

In general, any set of $t$ polynomials $\{\gox,\ldots,\gtx\}$ each of which has degree at most $r-1$ and such that 
$\rankq\{\goas,\ldots,\gtas\} = t$ can be used to repair $\fas$. The repair bandwidth in this case equals $b = \sum_{\alpha \in A \setminus \{\alpha^* \}} b_\alpha$ sub-symbols, 
where $b_\alpha \define \rankq\big(\{\goa,\ldots,\gta\}\big)$. 
Moreover, in this case $b_\alpha$ equals the number of sub-symbols downloaded during the repair process from the node storing $\fa$.  

\section{A Lower Bound on the Repair Bandwidth}
\label{sec:lower_bound}

In order to show that the scheme discussed in Section~\ref{sec:GW} is 
optimal, Guruswami and Wootters~\cite{GuruswamiWootters2016} established a lower bound on the repair bandwidth for Reed-Solomon codes \footnote{The derivations in the bound contained a minor error, which we correct in our derivation.}. We start our exposition by improving their bound. The result of this derivation also suggests the number of sub-symbols that need to be downloaded from each available node using an optimal repair scheme. Consequently, the bound allows one to perform a theoretical/numerical search for optimal repair schemes in a simplified manner.  

\begin{proposition} 
\label{pro:lower_bound}
Any linear repair scheme for Reed-Solomon codes $\rsk$ over the extension field $F = \fqt$ that uses the subfield $B = \fq$ 
requires a bandwidth of at least \vspace{-5pt}
\[
\ell \lfloor b_{\text{AVE}} \rfloor + (n-1-\ell)\lceil b_{\text{AVE}} \rceil \vspace{-5pt}
\] 
sub-symbols over $B$, where $n = |A| \leq |F|$, and where $b_{\text{AVE}}$ and $\ell$
are defined as \vspace{-5pt} 
\[
b_{\text{AVE}} \define \log_q\Big(\frac{(n-1)|F|}{(r-1)(|F|-1)+(n-1)}\Big), \vspace{-5pt}
\]
and $\ell \define n-1$ if $b_{\text{AVE}} \in \mathbb{Z}$, and \vspace{-5pt}
\[
\ell \define \left\lfloor \frac{L - (n-1)q^{-\lceil b_{\text{AVE}} \rceil}}{q^{-\lfloor b_{\text{AVE}} \rfloor} - q^{-\lceil b_{\text{AVE}} \rceil}}\right\rfloor
\vspace{-5pt}
\]
otherwise. Here,  \vspace{-5pt}
\[
L \define \frac{(r-1)(|F|-1)+(n-1)}{|F|}.
\]
\end{proposition} 
\begin{proof} 
The first part of the proof proceeds along the same lines as the proof of~\cite[Thm.~6]{GuruswamiWootters2016}.
But once the optimization problem is solved to arrive at a \emph{fractional} lower bound, rather than allowing
the number of sub-symbols downloaded from each available node to be real-valued, we perform a rounding procedure which leads to an improved \emph{integral} lower bound.  
 
Fix any $\alpha^* \in A$ and consider an arbitrary exact linear repair scheme of Reed-Solomon codes for the node storing $\fas$ that uses $b$ sub-symbols from $B$. By~\cite[Thm.~4]{GuruswamiWootters2016}, there is a set of $t$ polynomials $\gox,\ldots,\gtx$ such that $\rankq\big(\{\goas,\ldots,\gtas\}\big) = t$ and $\rankq\big(\{\goa,\ldots,\gta\}\big) = b_\alpha$, for all $\alpha \in A \setminus \{\alpha^*\}$, where
$b = \sum_{\alpha \in A \setminus \{\alpha^*\}}b_\alpha$.  
For each $\alpha \in A$, let \vspace{-5pt}
\[
S_\alpha \define \{\ba = (a_1,\ldots,a_t) \in B^t \colon \sum_{i=1}^t a_i\gia = 0\}.
\]
Since $\rank(\{\goa,\ldots,\gta\}) = b_\alpha$, we deduce that 
$\dim_B(S_\alpha) = t - b_\alpha$. Averaging over all \emph{nonzero} $\ba \in B^t$, 
we have \vspace{-5pt}
\begin{multline}
\label{eq:average}
\frac{1}{|F|-1} \sum_{\ba \in B^t \setminus \{\bO\}} 
|\{\alpha \in A \setminus \{\alpha^*\} \colon \ba \in S_\alpha\}|\\
= \frac{1}{|F|-1} \sum_{\alpha \in A \setminus \{\alpha^*\}} 
|\{\ba \in B^t \setminus \{\bO\} \colon \ba \in S_\alpha\}|\\
= \frac{1}{|F|-1} \sum_{\alpha \in A \setminus \{\alpha^*\}} 
(q^{t-b_\alpha} - 1) =: E.
\end{multline}
(Note that we added the correction term ``$-1$'' in the last sum of \eqref{eq:average} that was missing in the original proof of~\cite[Thm.~6]{GuruswamiWootters2016}.)
Therefore, there exists some $\ba^* \in B^t \setminus \{\bO\}$ so that $|\{\alpha \colon
\ba^* \in S_\alpha\}| \geq E$. Let $\gsx \define \sum_{i=1}^t a^*_i \gix$. 
By the choice of $\ba^*$, $\gsx$ vanishes on at least $E$ points of $A \setminus \{\alpha^*\}$. Also, since $\ba^* \neq \bO$, $\gsas = \sum_{i=1}^t a^*_i \gias \neq 0$. Therefore, $\gsx$ corresponds to a nonzero codeword in the dual code $\Cd$ and can hence have at most $r - 1$ roots. Thus,
\[
\frac{1}{|F|-1} \sum_{\alpha \in A \setminus \{\alpha^*\}} 
(q^{t-b_\alpha} - 1) = E \leq r-1,
\] 
or equivalently, 
\begin{equation} 
\label{eq:feasible}
\sum_{\alpha \in A \setminus \{\alpha^*\}} q^{-b_\alpha}
\leq \big((r-1)(|F|-1)+(n-1)\big) / |F| =: L.
\end{equation} 
Let \vspace{-10pt}
\begin{equation}
\label{eq:optimization} 
b_{\min} \define \min_{b_\alpha \in \{0,1,\ldots,t\}} \sum_{\alpha \in A \setminus \{\alpha^*\}}b_\alpha,\quad\ \text{subject to } \eqref{eq:feasible}.  
\end{equation} 
Then, any feasible repair scheme has to have $b \geq b_{\min}$. To solve the optimization problem~(\ref{eq:optimization}), the authors of~\cite[Thm.~6]{GuruswamiWootters2016} relaxed the condition that $b_{\alpha}$ are integer-valued and arrived at a lower bound that reads as $(n-1)b_{\text{AVE}}$, where $b_{\text{AVE}} \define \log_q\big((n-1)/L\big)$.
But one can still solve~\eqref{eq:optimization} for $b_{\alpha} \in \{0,1,\ldots,t\}$ and arrive at a closed form expression for $b_{\min}$.
To see how to accomplish this analysis, we first let $\{b_1,\ldots,b_{n-1}\}$ refer to $\{b_\alpha \colon \alpha \in A \setminus \{\alpha^*\}\}$. We then 
claim that  \vspace{-5pt}
\[
b^*_1 = \cdots = b^*_\ell = \lfloor b_{\text{AVE}} \rfloor, 
b^*_{\ell+1} = \cdots = b^*_{n-1} = \lceil b_{\text{AVE}} \rceil,
\]
where $\ell$ is the largest integer satisfying $\sum_{i=1}^{n-1}q^{-b^*_i} \leq L$, is an optimal solution of \eqref{eq:optimization}. To this end, if $(b_1,\ldots,b_{n-1})$ is an optimal solution of \eqref{eq:optimization}, and $b_i - b_j \geq 2$ for some $i$ and $j$, we may decrease $b_i$ by one and increase $b_j$ by one, and retain an 
optimal solution. Repeating this ``balancing'' procedure for as many times as possible, we obtain
an optimal solution for which $|b_i - b_j| \leq 1$, $i,j\in [n-1]$.
If $\min_i b_i < \lfloor b_{\text{AVE}} \rfloor$ then $(b_1,\ldots,b_{n-1})$ cannot be a feasible solution. Therefore, $\min_i b_i \geq \lfloor b_{\text{AVE}} \rfloor$. 
Because of the way $\ell$ was chosen, we always have 
$\sum_{i=1}^{n-1}b_i \geq \sum_{i=1}^{n-1}b^*_i$, 
which establishes the optimality of $(b^*_1,\ldots,b^*_{n-1})$.
Finally, $\ell$ may be easily computed as follows. 
If $b_{\text{AVE}} \in \mathbb{Z}$ then $\ell = n-1$, otherwise \vspace{-5pt}
\[
\ell = \left\lfloor \frac{L - (n-1)q^{-\lceil b_{\text{AVE}} \rceil}}{q^{-\lfloor b_{\text{AVE}} \rfloor} - q^{-\lceil b_{\text{AVE}} \rceil}}\right\rfloor. \qedhere
\]
\end{proof} 

\begin{corollary} 
\label{cr:bound1}
When $n = |F| = q^t$ and $r = q^s$, for some $s \in [t]$, 
any linear repair scheme over the subfield $B = \fq$ 
of a Reed-Solomon code $\rsk$ defined over $F$ requires a bandwidth of at least
$(n-1)(t-s)$ sub-symbols over $B$. 
\end{corollary}
\begin{proof} 
In this case, $b_{\text{AVE}} = t-s \in \mathbb{Z}$ and $\ell = n - 1$, which according to Proposition~\ref{pro:lower_bound} give the desired bound.  
\end{proof}  

Note also that the integral bound of Corollary~\ref{cr:bound1} and the Guruswami-Wootters fractional bound coincide. 
However, in many other cases, the integral bound strictly outperforms the fractional
bound. Consider as an example the Facebook RS(14,10) code defined over $\text{GF}(256)$. If the code is repaired over the subfield $\text{GF}(16))$, 
the fractional bound results in $28$ downloaded bits, while our integral bound asserts that a download of at least $44$ bits is needed. 
It is also apparent that the fractional bound does not depend on the subfield that the code is repaired over, while the integral bound does.  
In general, the bigger the order of $B$, the larger the gap between the two bounds. 

Also, one may assume that if a repair scheme that achieves the bound of Corollary~\ref{cr:bound1} were to exist, it would require that the 
replacement node download $t-s$ sub-symbols from each available
node. This intuition has been extremely useful in our quest for optimal repair schemes for Reed-Solomon codes.

\section{Optimal Repair Schemes for Full-Length Reed-Solomon Codes with $r = q^s$}
\label{sec:main}

In this section, we construct repair schemes for a Reed-Solomon code $\rsk$ defined over $F = \fqt$, where the code length $n = |A| \leq |F|$ and the number of parities $r \define n - k \geq q^s$, for every $s < t$.
These schemes are optimal when $n = |F| = q^t$ and $r = q^s$.  
We first settle the case $q = 2$ and $s = 1$, and then proceed to tackle the general case when $q \geq 2$ and $s < t$. 


\subsection{Repair Schemes for Reed-Solomon Codes with Two Parities}

Suppose that $q = 2$ and $r = n - k \geq 2$. 
We can use the constant and linear polynomials as check polynomials for repairing a codeword symbol $\fas$ (In fact, only linear polynomials are used).
The main task is to select the roots and multipliers of the codewords properly. 

\textbf{Construction I.}
Assume that $\alpha^* \in A$ and $\fas$ is erased. Select a subset $\{z_1,\ldots,z_t\}
\subseteq F$ such that $\{\alpha^* - z_1,\ldots,\alpha^* - z_t\}$ forms a basis of
$F$ over $\ft$. For $i \in [t],$ set $\be_i = \alpha^* - z_i$ and $\gix = \be_i(x-z_i)$.
  
\begin{lemma} 
\label{lem:ind}
The set $\{\goas,\ldots,\gtas\}$, where the $\gix$ are chosen according to Construction~I, has rank $t$ over $\ft$. 
\end{lemma} 
\begin{proof} 
We have $(\goas,\ldots,\gtas) = (\be_1^2,\ldots,\be_t^2)$. 
As we are working over a field of characteristic two, it holds that
$\sum_{i=1}^t a_i\be_i^2 = \big(\sum_{i=1}^t a_i \be_i\big)^2$, for every $a_i \in \ft$. 
Therefore, \vspace{-5pt}
\[
\sum_{i=1}^t a_i\be_i^2 = 0 \Longrightarrow \sum_{i=1}^t a_i \be_i = 0, \vspace{-5pt}
\] 
which implies that $a_1=\cdots=a_t = 0$, as $\{\be_1,\ldots,\be_t\}$ is a basis of $F$
over $\ft$. Therefore, $\rankt\big(\{\goas,\ldots,\gtas\}\big) = t$.  
\end{proof} 

\begin{lemma} 
\label{lem:dep}
If the $\gix$ are chosen according to Construction~I, then $\rankt\big(\{\goa,\ldots,\gta\}\big) \leq t-1$, for every $\alpha \in A \setminus \{\alpha^*\}$. 
\end{lemma} 
\begin{proof} 
We aim to show that the set $\{\goa,\ldots,\gta\}$ is dependent over $\ft$, for every
$\alpha \in A \setminus \{\alpha^*\}$. For $i \in [t]$, we have
\[
\gia = \be_i(\alpha - z_i) = \be_i\big((\alpha^*-z_i) - (\alpha^*-\alpha)\big)
= \be_i\big(\be_i - (\alpha^*-\alpha)\big).
\]
As $\{\be_1,\ldots,\be_t\}$ is a basis of $F$ over $\ft$, we may write  \vspace{-5pt}
\[
\alpha^*-\alpha = \sum_{i = 1}^t a_i \be_i, 
\]
where $a_i \in \ft$, $i \in [t]$, are not all zero. We now have
\vspace{-5pt}
\[
\begin{split}
\sum_{i=1}^t a_i \gia &= \sum_{i=1}^t a_i \be_i\big(\be_i - (\alpha^*-\alpha)\big)\\
&= \sum_{i=1}^t a_i \be_i^2 - \Big(\sum_{i=1}^t a_i \be_i\Big)(\alpha^* - \alpha)\\
&= \sum_{i=1}^t a_i \be_i^2 - \Big(\sum_{i=1}^t a_i \be_i\Big)\Big(\sum_{i=1}^t a_i \be_i\Big) = 0.
\end{split}
\] 
Therefore, the set $\{\goa,\ldots,\gta\}$ is dependent over $\ft$. 
\end{proof} 

\begin{theorem} 
\label{thm:r=2}
Let $n \leq |F| = 2^t$ and $r = n - k \geq 2$. 
The set of check polynomials $\{\gox,\ldots,\gtx\}$ defined in Construction~I can be used to repair a codeword symbol $\fas$ of a Reed-Solomon code $\rsk$ with a repair bandwidth of at most $(n-1)(t-1)$ bits. 
Moreover, when $n = |F| = 2^t$ and $r = 2$, this repair bandwidth is optimal.  
\end{theorem} 
\begin{proof} 
The first claim follows from Lemma~\ref{lem:ind} and Lemma~\ref{lem:dep}. 
The second claim holds due to Corollary~\ref{cr:bound1}, with $B = \ft$. 
\end{proof} 
An example illustrating Construction~I is given in Fig.~\ref{fig:ConstructionI}. 
\begin{table}[htb]
\centering
\begin{tabular}{|l||c|c|c|c|c|c|c|c|}
\hline
$A$ & $\boldsymbol{\alpha^*} = \boldsymbol{0}$ & $1$ & $\xi$ & $\xi^2$ & $\xi^3$ & $\xi^4$ & $\xi^5$ & $\xi^6$\\
\hhline{|=||=|=|=|=|=|=|=|=|}
$g_1 = x-1$ & $\mathbf{1}$ & $\cdot$ & $\xi^3$ & $\xi^6$ & $\xi$ & $\xi^5$ & $\xi^4$ & $\xi^2$\\
\hline
$g_2 = \xi(x-\xi)$ & $\mathbf{\boldsymbol{\xi}^2}$ & $\xi^4$ & $\cdot$ & $\xi^5$ & $\xi$ & $\xi^3$ & $1$ & $\xi^6$\\
\hline
$g_3 = \xi^2(x-\xi^2)$ & $\mathbf{\boldsymbol{\xi}^4}$ & $\xi$ & $\xi^6$ & $\cdot$ & $1$ & $\xi^3$ & $\xi^5$ & $\xi^2$\\
\hhline{|=||=|=|=|=|=|=|=|=|}
$\rankt(\cdot)$ & $\mathbf{3}$ & $2$ & $2$& $2$& $2$& $2$& $2$& $2$\\
\hline
\end{tabular}
\caption{The list of dual codewords generated according to Construction~I, which may be used to repair the first codeword symbol $f(0)$ of an $[8,6]$ Reed-Solomon code over $\mathbb{F}_{2^3}$. 
We let $\xi$ be a primitive element of the field, where $1 + \xi + \xi^3 = 0$. 
The column corresponding to the evaluation point $\alpha^* = 0$ has rank three,
which means that the corresponding dual codewords can be used to repair $f(0)$. All other columns have rank two over $\ft$, which means that this scheme has a repair bandwidth of $14 = 7*2$ bits, which is optimal. It suffices for the replacement node to download two bits from each available node, for instance, $\tr(\xi^4f(\xi^5))$ and $\tr(f(\xi^5))$ from the node storing $f(\xi^5)$, or $\tr(\xi^2f(\xi^6))$ and $\tr(\xi^6f(\xi^6))$ from the node storing $f(\xi^6)$.}
\label{fig:ConstructionI} \vspace{-15pt}
\end{table}

\subsection{Repair Scheme for Reed-Solomon Codes with $q^s$ Parities}

Suppose that $r = n - k \geq q^s$. 
We can use polynomials of degrees at most $r-1 = q^s-1$ as check polynomials for repairing a codeword symbol $\fas$, but we choose only polynomials of highest degree possible, i.e. of degree $q^s-1$. We generalize Construction~I by using the \emph{inverses} of the nonzero elements of a subspace of dimension $s$ to generate the check polynomials.

\textbf{Construction II.} 
Let $\xi$ be a primitive element of $F = \fqt$ and $W = \{0, w_1,\ldots,w_{q^s-1}\}$ an $\fq$-subspace of dimension $s$ in $\fqt$. 
For all $i \in [t]$, we choose
\[
\gix = \be_i\prod_{j=1}^{q^s-1} \Big(x - \big(\alpha^* - w_j^{-1} \be_i\big)\Big),
\]
where $\{\be_1,\ldots,\be_t\}$ is an arbitrary $\fq$-basis of $\fqt$. 
Note that Construction~I corresponds to the case $q = 2$, $s = 1$, and $w_1 = 1$. 
We set 
\begin{equation} 
\label{eq:wi}
M \define \prod_{j=1}^{q^s-1} w_j^{-1} \neq 0. 
\end{equation} 

\begin{lemma} 
\label{lem:qs_high}
The set $\{\goas,\ldots,\gtas\}$, where $\gix$ is described in Construction~II, has rank $t$ over $\fq$, for every $s < t$. 
\end{lemma} 
\begin{proof}
For $i \in [t]$ we have
\[
\gias = \be_i \prod_{j=1}^{q^s-1} \Big(\alpha^* - \big(\alpha^* - w_j^{-1} \be_i\big)\Big)
= M\be_i^{q^s}.
\]
It is obvious that $\{\be_1^{q^s},\ldots,\be_t^{q^s}\}$ is also an $\fq$-basis of $\fqt$. Since $M$ is a nonzero constant, we deduce that 
\[
\rankq\big(\{\goas,\ldots,\gtas\}\big) = \rankq\big(\{\be_1^{q^s},\ldots,\be_t^{q^s}\}\big) = t. \qedhere
\]
\end{proof} 

The proof of the following lemma can be found in~\cite[p.~4]{Goss}.

\vskip 5pt
\begin{lemma} 
\label{lem:linearized}
Suppose that $s < t$ and that $W$ is an $s$-dimensional $\fq$-subspace of $\fqt$. Let $L_W(x) = \prod_{w \in W}(x-w)$. Then $L_W$ is an $\fq$-linear mapping from $\fqt$
to itself, with kernel $W$ and image $L_W(\fqt)$ of dimension $t-s$
over $\fq$. \vspace{-5pt}
\end{lemma} 

\begin{lemma} 
\label{lem:qs_low}
The set $\{\goa,\ldots,\gta\},$ with the polynomials $\gix$ defined in Construction~II, 
has rank at most $t-s$ over $\fq$, for every $\alpha \in A \setminus \{\alpha^*\}$ and $s < t$.
\end{lemma} 
\begin{proof} 
Set $\gamma = \alpha^* - \alpha \neq 0$, we have
\begin{equation} 
\label{eq:pia}
\begin{split}
\gia &= \be_i\prod_{j=1}^{q^s-1} \Big(\alpha - \big(\alpha^* - w_j^{-1} \be_i\big)\Big)\\
&= \be_i\prod_{j=1}^{q^s-1} \Big(w_j^{-1} \be_i - \gamma\Big)\\
&= \be_i\gamma^{q^s-1}\prod_{j=1}^{q^s-1} \Big(w_j^{-1} \be_i\gamma^{-1} - 1\Big)\\
&= \gamma^{q^s}\big(\be_i\gamma^{-1}\big)\prod_{j=1}^{q^s-1} \Big(w_j^{-1} \big(\be_i\gamma^{-1}\big) - 1\Big). 
\end{split} 
\end{equation} 
If we set $v(x) = x\prod_{j=1}^{q^s-1}(w_j^{-1} x - 1)$, then \vspace{-5pt}
\begin{multline} 
\label{eq:vx} 
v(x) = x\Big(\prod_{j=1}^{q^s-1}w_j^{-1}\Big)\Big(\prod_{j=1}^{q^s-1}(x - w_j)\Big)\\ 
= M\Big(x\prod_{j=1}^{q^s-1}(x - w_j)\Big) = ML_W(x), 
\end{multline} 
due to \eqref{eq:wi}, where $L_W(x) \define \prod_{w \in W}(x-w)$.  
From \eqref{eq:pia} and \eqref{eq:vx}, we obtain
$\gia = M\gamma^{q^s} L_W\big(\be_i\gamma^{-1}\big)$.  
Therefore, 
\[
\begin{split} 
\rankq\big(\{\gia \colon i \in [t]\}\big) 
&= \rankq\big(\left\{L_W\big(\big(\be_i\gamma^{-1}\big)\big) \colon i \in [t]\right\}\big)\\
&\leq \dim_q\big(L_W(\fqt)\big) = t-s,
\end{split} 
\]
where the last inequality follows from Lemma~\ref{lem:linearized}. 
\end{proof} 

\begin{theorem} 
\label{thm:r=qs2} 
The statements of Theorem~\ref{thm:r=qs} hold for the set of check polynomials defined in Construction~II. 
\end{theorem} 

The next construction achieves the same repair bandwidth as Construction~II, via the nonzero elements of a subspace of dimension $s$. In fact, this construction generalizes the construction of~\cite{GuruswamiWootters2016} by using a \emph{linearized polynomial} with distinct roots (see, for instance~\cite[Ch.~4, \S9]{MW_S}) instead of the field trace polynomial. Note that the field trace $\tr_{\fqt/\fqts}$ is well defined only when $\fqts$ is a \emph{subfield} of $\fqt$, i.e. when $(t-s) | t$.
In contrast, a linearized polynomial of degree $q^s$ with no repeated roots, which maps $\fqt$ to a \emph{subspace} of dimension $t-s$, always exists for every $1 \leq s < t$.  

\textbf{Construction III.}
Let $\{u_1,\ldots,u_t\}$ be an $\fq$-basis of $\fqt$ and let $W$ be an arbitrary $\fq$-subspace of dimension $s$ of $\fqt$. Set $L_W(x) \define \prod_{w \in W}(x-w)$ and $g_i(x) \define L_W\big(u_i(x-\alpha^*)\big)/(x-\alpha^*)$, for every $i \in [t]$. Note that since $\deg(g_i) = q^s-1 \leq r-1$, the polynomials $g_i(x)$ are checks for the Reed-Solomon code $\C$. 

\begin{lemma} 
\label{lem:linearized_high}
The set $\{\goas,\ldots,\gtas\}$, where $\gix$ is described in Construction~III, has rank $t$ over $\fq$, for every $s < t$. 
\end{lemma}  
\begin{proof} 
Let $W^* \define W \setminus \{0\}$. Then $L_W(x) = x\prod_{w \in W^*}(x-w) = \tau x + x^2h(x)$, where $\tau = (-1)^{q^s-1} \prod_{w \in W^*}w \neq 0$ and $h(x)$ is a polynomial of degree $q^s-2$. 
Therefore, $g_i(x) = \tau u_i + u_i^2(x-\alpha^*)h\big(u_i(x-\alpha^*)\big)$ and hence, $g_i(\alpha^*) = \tau u_i$, for every $i \in [t]$. As $\tau \neq 0$ and $\{u_1,\ldots,u_t\}$ is an $\fq$-basis of $\fqt$, it follows that the set $\{\goas,\ldots,\gtas\}$ has rank $t$ over $\fq$.
\end{proof} 
\vskip 5pt

\begin{lemma} 
\label{lem:linearized_low}
The set $\{\goa,\ldots,\gta\}$ with the polynomials $\gix$ defined in Construction~III, 
has rank at most $t-s$ over $\fq$, for every $\alpha \in A \setminus \{\alpha^*\}$ and $1 \leq s < t$.
\end{lemma} 
\begin{proof} 
For $\alpha \neq \alpha^*$, set $\gamma_i = u_i(\alpha - \alpha^*)$, we have $g_i(\alpha) = \frac{1}{\alpha - \alpha^*}L_W\big(\gamma_i\big)$. Therefore \vspace{-5pt}
\[
\begin{split}
\rankq\big(\{\goa,\ldots,\gta\}\big) &= \rankq\big(\left\{L_W(\gamma_1),\ldots, L_W(\gamma_t)\right\}\big)\\
&\leq \dim_q\big(L_W(\fqt)\big) = t - s, \vspace{-5pt} 
\end{split}
\]
according to Lemma~\ref{lem:linearized}. 
\end{proof} 

\begin{theorem} 
\label{thm:r=qs}
Let $n \leq |F| = q^t$ and $r = n - k \geq q^s$, for some $s < t$. 
The set of check polynomials $\{\gox,\ldots,\gtx\}$ defined in Construction~III
can be used to repair a codeword symbol $\fas$ of a Reed-Solomon code $\rsk$ with a repair bandwidth of at most $(n-1)(t-s)\log_2q$ bits. 
Moreover, when $n = |F| = q^t$ and $r = q^s$, this repair bandwidth is optimal.  
\end{theorem} 
\begin{proof} 
The first statement follows from Lemma~\ref{lem:linearized_high} and Lemma~\ref{lem:linearized_low}. 
The second statement holds due to Corollary~\ref{cr:bound1}. 
\end{proof} 

\section*{Acknowledgment}

This work has been supported in part by the NSF grant CCF 1526875 and the Center for Science of Information under the grant NSF 0939370.
The authors thank Iwan Duursma for helpful discussions.  

\bibliographystyle{IEEEtran}
\bibliography{FullLength}

\begin{thebibliography}{10}
\providecommand{\url}[1]{#1}
\csname url@samestyle\endcsname
\providecommand{\newblock}{\relax}
\providecommand{\bibinfo}[2]{#2}
\providecommand{\BIBentrySTDinterwordspacing}{\spaceskip=0pt\relax}
\providecommand{\BIBentryALTinterwordstretchfactor}{4}
\providecommand{\BIBentryALTinterwordspacing}{\spaceskip=\fontdimen2\font plus
\BIBentryALTinterwordstretchfactor\fontdimen3\font minus
  \fontdimen4\font\relax}
\providecommand{\BIBforeignlanguage}[2]{{%
\expandafter\ifx\csname l@#1\endcsname\relax
\typeout{** WARNING: IEEEtran.bst: No hyphenation pattern has been}%
\typeout{** loaded for the language `#1'. Using the pattern for}%
\typeout{** the default language instead.}%
\else
\language=\csname l@#1\endcsname
\fi
#2}}
\providecommand{\BIBdecl}{\relax}
\BIBdecl

\bibitem{Dimakis_etal2007}
A.~Dimakis, P.~Godfrey, M.~Wainwright, and K.~Ramchandran, ``Network coding for
  distributed storage systems,'' in \emph{Proc. 26th IEEE Int. Conf. Comput.
  Commun. (INFOCOM)}, 2007, pp. 2000--2008.

\bibitem{Dimakis_etal2010}
A.~Dimakis, P.~Godfrey, Y.~Wu, M.~Wainwright, and K.~Ramchandran, ``Network
  coding for distributed storage systems,'' \emph{IEEE Trans. Inform. Theory},
  vol.~56, no.~9, pp. 4539--4551, 2010.

\bibitem{ReedSolomon1960}
I.~S. Reed and G.~Solomon, ``Polynomial codes over certain finite fields,''
  \emph{J. Soc. Ind. Appl. Math.}, vol.~8, no.~2, pp. 300--304, 1960.

\bibitem{MW_S}
F.~J. MacWilliams and N.~J.~A. Sloane, \emph{The Theory of Error-Correcting
  Codes}.\hskip 1em plus 0.5em minus 0.4em\relax Amsterdam: North-Holland,
  1977.

\bibitem{Dimakis_etal2010_survey}
A.~Dimakis, K.~Ramchandran, Y.~Wu, and C.~Suh, ``A survey on network codes for
  distributed storage,'' \emph{Proc. IEEE}, vol.~99, no.~3, pp. 476--489, 2011.

\bibitem{OggierDatta2011}
F.~Oggier and A.~Datta, ``Self-repairing homomorphic codes for distributed
  storage systems,'' in \emph{Proc. IEEE Int. Conf. Comput. Commun. (INFOCOM)},
  2011, pp. 1215--1223.

\bibitem{GopalanHuangSimitciYekhanin2012}
P.~Gopalan, C.~Huang, H.~Simitci, and S.~Yekhanin, ``On the locality of
  codeword symbols,'' \emph{IEEE Trans. Inform. Theory}, vol.~58, no.~11, pp.
  6925--6934, 2012.

\bibitem{PapailiopoulosDimakis2012}
D.~Papailiopoulos and A.~Dimakis, ``Locally repairable codes,'' in \emph{IEEE
  Int. Symp. Inform. Theory (ISIT)}, 2012, pp. 2771--2775.

\bibitem{GuruswamiWootters2016}
V.~Guruswami and M.~Wootters, ``Repairing {R}eed-{S}olomon codes,'' in
  \emph{Proc. Annu. Symp. Theory Comput. (STOC)}, 2016.

\bibitem{WikiEC}
``Erasure coding for distributed storage {W}iki,'' available at
  \url{http://storagewiki.ece.utexas.edu/doku.php?id=wiki:papers:all}.

\bibitem{RashmiShahKumar2011}
K.~V. Rashmi, N.~B. Shah, and P.~V. Kumar, ``Optimal exact-regenerating codes
  for distributed storage at the {MSR} and {MBR} points via a product-matrix
  construction,'' \emph{IEEE Trans. Inform. Theory}, vol.~57, no.~8, pp.
  5227--5239, 2011.

\bibitem{TamoWangBruck2011}
I.~Tamo, Z.~Wang, and J.~Bruck, ``Zigzag codes: {MDS} array codes with optimal
  rebuilding,'' {T}ech. {R}ep., available
  at\url{http://www.paradise.caltech.edu/papers/etr110.pdf}.

\bibitem{WangTamoBruck2011}
Z.~Wang, I.~Tamo, and J.~Bruck, ``On codes for optimal rebuilding access,'' in
  \emph{Proc. 49th Annual Allerton Conf. Comm Control Comput. (Allerton)},
  2011, pp. 1374--1381.

\bibitem{Shanmugam2014}
K.~Shanmugam, D.~S. Papailiopoulos, A.~G. Dimakis, and G.~Caire, ``A repair
  framework for scalar {MDS} codes,'' \emph{IEEE J. Selected Areas Comm.
  (JSAC)}, vol.~32, no.~5, pp. 998--1007, 2014.

\bibitem{DauDuursmaKiahMilenkovic2016}
H.~Dau, I.~Duursma, H.~M. Kiah, and O.~Milenkovic, ``Repairing {R}eed-{S}olomon
  codes with multiple erasures,'' available at
  \url{https://arxiv.org/abs/1612.01361}.

\bibitem{YeBarg_ISIT2016}
M.~Ye and A.~Barg, ``Explicit constructions of {MDS} array codes and {RS} codes
  with optimal repair bandwidth,'' in \emph{IEEE Int. Symp. Inform. Theory
  (ISIT)}, 2016, pp. 1202--1206.

\bibitem{LidlNiederreiter1986}
R.~Lidl and H.~Niederreiter, \emph{Introduction to Finite Fields and Their
  Applications}.\hskip 1em plus 0.5em minus 0.4em\relax Cambridge University
  Press, 1986.

\bibitem{Goss}
D.~Goss, \emph{Basic Structures of Function Field Arithmetics}.\hskip 1em plus
  0.5em minus 0.4em\relax Springer-Verlag Berlin Heidelberg, 1996.

\end{thebibliography}

\end{document}